\documentclass[final]{siamltex}

 \usepackage{graphicx,fancybox,latexsym,epsfig}
\usepackage{fancyhdr,amsmath,times,amsxtra,amssymb}
\usepackage{color}

\usepackage{dsfont}
\usepackage{amsmath}

\def\Pr{\mathop{\rm Pr}}

\newtheorem{assumption}{Assumption}

\newtheorem{exmp}{Example}
\newtheorem{mydef}{Definition}

\newtheorem{remark}{Remark}

\newcommand{\Zplus}{\mathds{Z}_+}

\newcommand{\R}{\mathds{R}}
\newcommand{\N}{\mathds{N}}

\begin{document}

\sloppy
\title{Robustness to Incorrect Models and Data-Driven Learning in Average-Cost Optimal Stochastic Control
\thanks{
This research was supported in part by
the Natural Sciences and Engineering Research Council (NSERC) of Canada. Some preliminary results involving this paper were presented at the 2019 IEEE Conference on Decision and Control \cite{karaCDC2019robustness}.}
}
\author{Ali Devran Kara, Maxim Raginsky, and Serdar Y\"uksel
\thanks{A. D. Kara and S. Y\"uksel are with the Department of Mathematics and Statistics,
     Queen's University, Kingston, ON, Canada,
     email: \{16adk,yuksel\}@mast.queensu.ca. M. Raginsky is with the University of Illinois, email: maxim@illinois.edu}
     }
\maketitle

\begin{abstract}
We study continuity and robustness properties of infinite-horizon average expected cost problems with respect to (controlled) transition kernels, and applications of these results to the problem of robustness of control policies designed for approximate models applied to actual systems. We show that sufficient conditions presented in the literature for discounted-cost problems are in general not sufficient to ensure robustness for average-cost problems. However, we show that the average optimal cost is continuous in the convergences of controlled transition kernel models where convergence of models entails (i) continuous weak convergence in state and actions, and (ii) continuous setwise convergence in the actions for every fixed state variable, in addition to either uniform ergodicity or some regularity conditions. We establish that the mismatch error due to the application of a control policy designed for an incorrectly estimated model to the true model decreases to zero as the incorrect model approaches the true model under the stated convergence criteria. Our findings significantly relax related studies in the literature which have primarily considered the more restrictive total variation convergence criteria. Applications to robustness to models estimated through empirical data (where almost sure weak convergence criterion typically holds, but stronger criteria do not) are studied and conditions for asymptotic robustness to data-driven learning are established.
\end{abstract}

\section{Introduction}\label{section:intro}

\subsection{Preliminaries}
The paper studies continuity and robustness properties of infinite-horizon average cost problems with respect to transition probabilities. Continuity results are used to establish robustness of optimal control policies applied to systems with incorrect or approximate models. Before proceeding further with the problem formulation, we present the setup. Let $\mathds{X} \subset \mathds{R}^n$ be a Borel set in which the elements of a
controlled Markov process $\{X_t,\, t \in \Zplus\}$ take values.  Here
and throughout the paper, $\Zplus$ denotes the set of non-negative
integers and $\mathds{N}$ denotes the set of positive integers. Let $\mathds{U}$, the action space, be a Borel
subset of some Euclidean space. An {\em admissible policy} $\gamma$ is a
sequence of control functions $\{\gamma_t,\, t\in \Zplus\}$, such
that $\gamma_t$ is measurable on the $\sigma$-algebra
generated by the information variables
\[
I_t=\{X_{[0,t]},U_{[0,t-1]}\}, \quad t \in \mathds{N}, \quad
  \quad I_0=\{X_0\},
\]
where
\begin{equation}
\label{eq_control}
U_t=\gamma_t(I_t),\quad t\in \Zplus,
\end{equation}
are the $\mathds{U}$-valued control
actions.
\noindent We define $\Gamma$ to be the set of all such admissible policies.

The joint distribution of the state and control
processes is determined by (\ref{eq_control}) and the following
relationship:
\begin{eqnarray*}
\label{eq_evol}
& \Pr\biggl( X_t\in B \, \bigg|\, (X,U)_{[0,t-1]}=(x,u)_{[0,t-1]} \biggr) \\
& \qquad = \int_B \mathcal{T}( dx_t|x_{t-1}, u_{t-1}),  B\in \mathcal{B}(\mathds{X}), t\in \mathds{N},
\end{eqnarray*}
where $\mathcal{T}(\cdot|x,u)$ is a stochastic kernel  (that is, a regular conditional probability measure) from $\mathds{X}\times
\mathds{U}$ to $\mathds{X}$.


The objective of the controller is to minimize the infinite-horizon average expected cost
  \begin{align*}
    J_{\infty}({\mathcal{T}},\gamma)= \limsup_{N\to \infty}\frac{1}{N}E_{x_0}^{{\mathcal{T}},\gamma}\left[\sum_{t=0}^{N-1} c(X_t,U_t)\right]
  \end{align*}
over the set of admissible policies $\gamma\in\Gamma$, where $c:\mathds{X}\times\mathds{U}\to\R$ is the stage-wise (Borel measurable) cost function and $E_{x_0}^{{\mathcal{T}},\gamma}$ denotes the expectation with initial state $x_0$ and transition kernel $\mathcal{T}$ under policy $\gamma$. To denote the explicit dependence of the optimal cost in the transition kernel, we use the notation
\begin{align*}
  J_{\infty}^*({\mathcal{T}})&=\inf_{\gamma\in\Gamma} J_{\infty}(\mathcal{T},\gamma).
\end{align*}
The focus of the paper is to address the following problems:

\noindent{\bf Problem P1: Continuity of  $J_\infty^*(\mathcal{T})$ (optimal cost for controlled setup) in transition kernels.}
 Suppose ${\{\mathcal{T}_n, n \in \mathds{N}\}}$ is a
sequence of transition kernels converging in some sense to $\mathcal{T}$. When does $\mathcal{T}_n \to \mathcal{T}$
imply
\begin{align*}
  J_\infty^*(\mathcal{T}_n)&\to J_\infty^*(\mathcal{T})?
\end{align*}
\noindent{\bf Problem P2: Robustness of  policies designed for incorrectly estimated models.}
 Suppose ${\{\mathcal{T}_n, n \in \mathds{N}\}}$ is a
sequence of transition kernels converging in some sense to $\mathcal{T}$. If we design optimal policies $\gamma_n^*$ for the estimated model $\mathcal{T}_n$ and apply them to the real model $\mathcal{T}$, when does $\mathcal{T}_n \to \mathcal{T}$
imply
\begin{align*}
  J_\infty(\mathcal{T},\gamma_n^*)\to J_\infty^*(\mathcal{T})?
\end{align*}
Application of the above to empirical learning, where models are estimated through empirical measurements, will also be studied.

\subsection{Literature Review and Contributions}
Robustness is a desired property for the optimal control of stochastic or deterministic systems when a given model does not reflect the actual system perfectly, as is usually the case in practice. This is a classical problem, so there is a very large literature on robust stochastic control and its application to learning-theoretic methods; see e.g. \cite{jacobson1973optimal,dupuis2000robust,dupuis2000kernel,dai1996connections,esfahani2015,erdogan2005,sun2015,basu2008stochastic,Kleptsyna2016,ugrinovskii1998robust}. 

In \cite{Kernel2018_cdc,kara2020robustness}, we studied continuity and robustness properties of partially and fully observed models under weak and total variation convergence of transition probabilities for infinite-horizon optimal stochastic control under discounted-cost criteria. 
We showed that the expected induced cost is robust under total variation in that it is continuous in the mismatch of transition kernels under convergence in total variation for the discounted cost setup. By imposing further assumptions on the measurement models and on the kernel itself, it was shown that the optimal discounted cost can be made continuous under weak convergence of transition kernels as well. For the expected average-cost criteria, in this paper, we show that the sufficient conditions presented in \cite{Kernel2018_cdc,kara2020robustness} for the expected discounted cost criteria may not guarantee continuity and robustness. This arises from the fact that the persistent errors between the transition probabilities for time stages in the distant future still matter unlike the discounted-cost setup, where such errors are discounted away. 



In our analysis, we will study various convergence criteria. Our findings involving weak convergence are related to the following convergence notion: Let a stochastic process converge to another one in the following sense: all finite dimensional marginals converge weakly and the conditional kernels on the future random variables given the past converge weakly as well. This has been termed as {\it extended weak convergence} \cite{aldous1981weak} or convergence under the \textit{the information topology} \cite{hellwig1996sequential}; these have recently shown to be equivalent in discrete-time \cite[Theorem 1.1]{backhoff2019all}. Applications of this notion of convergence to robustness to system models in stochastic control have recently been studied in different contexts \cite{bayraktar2020continuity,julio2020adapted}, \cite{kara2020robustness}. In our case, as in \cite{kara2020robustness}, we will relate this to a generalized (dominated) convergence criterion under {\it continuous weak convergence} \cite[Theorem 3.5]{Langen81} and \cite[Theorem 3.5]{serfozo82}. However, we emphasize that, unlike \cite{kara2020robustness}, the average cost setup requires additional novel arguments. Furthermore, convergence notions more general than weak convergence will be established; notably, setwise convergence of kernels in actions for each fixed state will be a setup to also be studied extensively in our paper. 

A particularly relevant study which has investigated the robustness problem under the expected average-cost criterion is \cite{her89} by Hernandez-Lerma: In \cite[Chapter 3]{her89}, a related problem to what we study in this paper is considered: Assume the transition probabilities are estimated through time, and at every time step $t$ the estimate is updated to some $\mathcal{T}_t$. Is it true that, as the estimates through time get closer to the true model (as $\mathcal{T}_t \to \mathcal{T}$) we can establish continuity and robustness? It is shown in \cite[Chapter 3]{her89} that, using some geometric ergodicity conditions and with uniform convergence  $\sup_{x,u}\|\mathcal{T}_n(\cdot|x,u)-\mathcal{T}(\cdot|x,u)\|_{TV}\to 0$, the answer to the question is positive. We will show that the conditions imposed in \cite[Chapter 3]{her89} are more restrictive than what we will present in this paper.

{\bf Contributions.} In view of the literature review reported above, we establish robustness and continuity results under much weaker conditions on the proximity and convergence properties between incorrect models and a true model: we show that the average optimal cost is continuous in the convergences of controlled transition kernel models where convergences of models entail (i) continuous weak convergence in state and actions, or (ii) continuous setwise convergence in the actions for every fixed state variable, in addition to ergodicity conditions. We show that the mismatch error due to the application of a control policy designed for an incorrectly estimated model to the true model decreases to zero as the incorrect model approaches the true model under the stated convergence criteria (Theorem \ref{weak_robust}, Theorem \ref{setwise_robust} and Theorem \ref{TV_robust}). We will also establish an ergodic invariant measure argument, in Theorem \ref{equ_pol}, for cases where the average cost optimality equation cannot be established. These findings are applied to empirical consistency results in data-driven stochastic control, where continuous weak convergence will be shown to be particularly relevant building on a measure concentration analysis. Compared to the results in \cite{her89} by Hernandez-Lerma, where a uniform total variation convergence over state and control action variables of the transition kernels is provided as a sufficient condition, we show that the uniform convergence on the state variable may be relaxed under the total variation convergence, and more importantly, we establish continuity and robustness under weak convergence and setwise convergences of the transition kernels which are much more relaxed notions of convergence. 


\section{Some Examples and Convergence Criteria for Transition Kernels}\label{def_exmp}

\subsection{Convergence criteria for transition kernels}
Before presenting convergence criteria for controlled transition kernels, we first review convergence of probability measures. Three important notions of convergences for sets of probability measures to be studied in the paper are weak convergence, setwise convergence and convergence under total variation. For $N\in\N$, a sequence $\{\mu_n,n\in\N\}$ in $\mathcal{P}(\R^N)$ is said to converge to $\mu\in\mathcal{P}(\R^N)$ \emph{weakly} if
  \begin{align}\label{converge}
    \int_{\R^N}c(x)\mu_n(dx) \to \int_{\R^N}c(x)\mu(dx)
  \end{align}
  \noindent for every continuous and bounded $c:\R^N \to \R$. $\{\mu_n\}$ is said to converge \emph{setwise} to $\mu\in\mathcal{P}(\R^N)$ if (\ref{converge}) holds for all measurable and bounded $c:\R^N \to \R$. 

  For probability measures $\mu,\nu\in\mathcal{P}(\R^N)$, the \emph{total variation} metric is given by
  \begin{align*}
    \|\mu-\nu\|_{TV}=\sup_{f:\|f\|_\infty \leq 1}\left|\int f(x)\mu(dx)-\int f(x)\nu(dx)\right|,
  \end{align*}
  \noindent where the supremum is taken over all measurable real-valued $f$, such that  $\|f\|_\infty=\sup_{x\in\R^N}|f(x)|\leq 1$. A sequence $\{\mu_n\}$ is said to converge in total variation to $\mu\in\mathcal{P}(\R^N)$ if $\|\mu_n-\mu\|_{TV}\to 0$.
Total variation defines a stringent metric for convergence; for example, a sequence of discrete probability measures does not converge in total variation to a probability measure which admits a density function. Setwise convergence, though, induces a topology on the space of probability measures which is not metrizable \cite[p.~59]{Ghosh}. However, the space of probability measures on a complete, separable, metric (Polish) space endowed with the topology of weak convergence is itself complete, separable and metric \cite{Par67}. Building on the above, we introduce the following convergence notions for (controlled) transition kernels.
\begin{mydef}\label{def_kernel}
For a sequence of transition kernels $\{ \mathcal{T}_n,n\in\mathds{N}\}$, we say that
\begin{enumerate}
\item $\mathcal{T}_n \to  \mathcal{T}$ weakly if $ \mathcal{T}_n(\cdot|x,u) \to  \mathcal{T}(\cdot|x,u)$ weakly, for all $x \in \mathds{X}$ and $u \in \mathds{U}$.
\item $\mathcal{T}_n \to  \mathcal{T}$ setwise if $ \mathcal{T}_n(\cdot|x,u) \to  \mathcal{T}(\cdot|x,u)$ setwise, for all $x \in \mathds{X}$ and $u \in \mathds{U}$.
\item $\mathcal{T}_n \to \mathcal{T}$ under the total variation distance if $ \mathcal{T}_n(\cdot|x,u) \to  \mathcal{T}(\cdot|x,u)$ under total variation, for all $x \in \mathds{X}$ and $u \in \mathds{U}$.
\end{enumerate}
\end{mydef}
We also note here that relative entropy convergence, through Pinsker's inequality \cite[Lemma 5.2.8]{GrayInfo}, is stronger than even total variation convergence which has also been studied in robust stochastic control as reviewed earlier. Another metric for probability measures is the Wasserstein distance: For compact spaces, the Wasserstein distance of order $1$, denoted by $W_1$, metrizes the weak topology (see \cite[Theorem 6.9]{villani2008optimal}). For non-compact spaces convergence in the $W_1$ metric implies weak convergence (in particular this metric bounds from above the Bounded-Lipschitz metric \cite[p.109]{villani2008optimal}). Considering these relations, our results in this paper can be directly generalized to the relative entropy distance or the Wasserstein metric. 

\subsection{Examples and Applications}\label{ornek1}
In the following we give some examples to show what these convergence types mean in terms of the functional representation of the dynamics. Some of the examples are taken from \cite{kara2020robustness} and presented here for completeness. Some applications in particular in the context of linear systems can be found in \cite{aastrom2013adaptive,kumar2015stochastic,goodwin1981discrete,goodwin2014adaptive}. Let a controlled model be given as
\[x_{t+1}=F(x_t,u_t,w_t),\] where $\{w_t\}$ is an i.i.d. noise process. The uncertainty on the transition kernel for such a system may arise from lack of information on $F$ or the i.i.d. noise process $w_t$:
\begin{itemize}
\item[(i)] Let $\{F_n\}$ denote an approximating sequence for $F$, so that $F_n(x,u,w) \to F(x,u,w)$ pointwise. Assume that the probability measure of the noise is known. Then, the corresponding kernels $\mathcal{T}_n$ converge weakly to $\mathcal{T}$: If we denote the probability measure of $w$ with $\mu$, for any $g \in C_b(\mathds{X})$, where $C_b(\mathds{X})$ denotes the continuous and bounded functions on $\mathds{X}$, and for any $(x_0,u_0)\in \mathds{X}\times\mathds{U}$ using the dominated convergence theorem we have
\begin{align*}
&\lim_{n \to \infty}\int g(x_1)\mathcal{T}_n(dx_1|x_0,u_0) \\
&=\lim_{n \to \infty}\int g(F_n(x_0,u_0,w)) \mu(dw)\\
&=\int g(F(x_0,u_0,w)) \mu(dw) = \int g(x_1)\mathcal{T}(dx_1|x_0,u_0).
\end{align*}

\item[(ii)] Much of the robust control literature deals with deterministic systems, where the actual model is a deterministic perturbation of the nominal model (see e.g. \cite{xie92,savkin1996robust,aastrom2013adaptive,goodwin1981discrete}). The considered model is in the following form; $\tilde{F}(x_t,u_t) = F(x_t,u_t) + \Delta F(x_t,u_t)$, where $F$ represents the nominal model and $\Delta F$ is the model uncertainty satisfying some norm bounds. For such deterministic systems, pointwise convergence of $\tilde{F}$ to the nominal model $F$, i.e. $\Delta F(x_t,u_t)\to 0$, can be viewed as weak convergence for deterministic systems by the discussion in (i). It is evident, however, that total variation convergence would be too strong for such a convergence criterion, since  $\delta_{\tilde{F}(x_t,u_t)} \to \delta_{F(x_t,u_t)}$ weakly but $\|\delta_{\tilde{F}(x_t,u_t)}-\delta_{F(x_t,u_t)}\|_{TV}=2$ for all $\Delta F(x_t,u_t)> 0$.

\item[(iii)]
Let $F(x_t,u_t,w_t) = f(x_t,u_t) + w_t$ be such that the function $f$ is known and the probability law $\mu$ of $w_t$ is misspecified, such that an incorrect model $\mu_n$ is assumed. If $\mu_n \to \mu$ weakly, setwise or in total variation, then the corresponding transition kernels $\mathcal{T}_n$ converges in the same sense to $\mathcal{T}$. Observe the following,
\begin{align}\label{exmp1}
&\int g(x_1)\mathcal{T}_n(dx_1|x_0,u_0) - \int g(x_1)\mathcal{T}(dx_1|x_0,u_0)\nonumber\\
&=\int g(w_0+f(x_0,u_0))\mu_n(dw_0) \nonumber\\
& \qquad - \int g(w_0+f(x_0,u_0))\mu(dw_0).
\end{align}

\begin{itemize}
\item[a.]Suppose $\mu_n \to \mu$ weakly. If $g$ is a continuous and bounded function then $g(\cdot + f(x_0,u_0))$ is a continuous and bounded function for all $(x_0,u_0)\in \mathds{X}\times\mathds{U}$. Thus, (\ref{exmp1}) goes to 0. Note that $f$ does not need to be continuous.

\item[b.]Suppose $\mu_n \to \mu$ setwise. If $g$ is a measurable and bounded function, then $g(\cdot +f(x_0,u_0))$ is measurable and bounded for all $(x_0,u_0)\in \mathds{X}\times\mathds{U}$. Thus, (\ref{exmp1}) goes to 0.

\item[c.]Finally, assume  $\mu_n \to \mu$ in total variation. If $g$ is bounded, (\ref{exmp1}) converges to 0, as in item (b). A special case for this would be the following: Assume $\mu_n$ and $\mu$ admit densities $h_n$ and $h$ respectively, then the pointwise convergence of $h_n$ to $h$ implies the convergence of $\mu_n$ to $\mu$ in total variation by Scheffe's Theorem \cite{BillingsleyProbMeasure}.  
\end{itemize}
\item[(v)]Suppose now neither $F$ nor the probability model of $w_t$ is known perfectly. It is assumed that $w_t$ admits a measure $\mu_n$ and $\mu_n\to \mu$ weakly. For the function $F$ we again have an approximating sequence $\{F_n\}$. If $F_n(x,u,w_n) \to F(x,u,w)$ for all $(x,u) \in\mathds{X}\times\mathds{U}$ and for any $w_n \to w$, then the transition kernel $\mathcal{T}_n$ corresponding to the model $F_n$ converges weakly to the one of $F$, $\mathcal{T}$: For any $g \in C_b(\mathds{X})$,
\begin{align*}
&\lim_{n \to \infty}\int g(x_1)\mathcal{T}_n( dx_1|x_0,u_0) =\lim_{n \to \infty}\int g(F_n(x_0,u_0,w)) \mu_n(dw)\\
& =\int g(F(x_0,u_0,w)) \mu(dw)= \int g(x_1)\mathcal{T}(dx_1|x_0,u_0). 
\end{align*}
In the analysis above we used a generalized dominated convergence result, Lemma \ref{langen}, to be presented later building on \cite[Theorem 3.5]{Langen81} and \cite[Theorem 3.5]{serfozo82}. 

\item[(vi)] Suppose that $F(x,u,\cdot): \mathds{W}\to \mathds{X}$ is invertible for all fixed $(x,u)$ and $F(x,u,w)$ is continuous and bounded on $\mathds{X}\times\mathds{U}\times\mathds{W}$. We construct the empirical measures for the noise process $w_t$  such that for every (fixed) Borel $B \subset \mathds{W}$, and for every $n \in\mathds{N}$, the empirical occupation measures are
\begin{align*}\label{inverse_emp}
\mu_n(B)=\frac{1}{n-1}\sum_{i=1}^{n} \mathds{1}_{\{F^{-1}_{x_{i-1},u_{i-1}}(x_i) \in B\} }
\end{align*}
where $F^{-1}_{x_{i-1},u_{i-1}}(x_i)$ denotes the inverse of $F(x_{i-1},u_{i-1},w): \mathds{W}\to \mathds{X}$ for given $(x_{i-1},u_{i-1})$. 
Using the noise measurements, we construct the empirical transition kernel estimates for any $(x_0,u_0)$ and Borel $B$ as
\begin{align*}
\mathcal{T}_n(B|x_0,u_0)=\mu_n(F^{-1}_{x_0,u_0}(B)).
\end{align*}
We have $\mu_n \to \mu$ weakly with probability one (\cite{Dudley02}, Theorem 11.4.1). This also implies that the transition kernels are such that $\mathcal{T}_n(\cdot|x_n,u_n)\to\mathcal{T}(\cdot|x,u)$ weakly for any $(x_n,u_n) \to (x,u)$. To see that observe the following for $h \in C_b(\mathds{X})$,
\begin{align*}
&\int h(x_1)\mathcal{T}_n(dx_1|x_n,u_n) - \int h(x_1)\mathcal{T}(dx_1|x,u)\\
&=\int h(f(x_n,u_n,w))\mu_n(dw) - \int h(f(x,u,w))\mu(dw)\to0.
\end{align*}
For the last step, we used that $\mu_n \to \mu$ weakly and $h(f(x_n,u_n,w))$ continuously converge to $h(f(x,u,w))$ i.e. $h(f(x_n,u_n,w_n)) \to h(f(x,u,w)$ for some $w_n \to w$ since $f$ and $h$ are continuous functions.


\end{itemize}


\section{Supporting Results: Continuity under Convergence of Transition Kernels}\label{TV_section}

\subsection{Some differences with the infinite horizon discounted problem}

 In \cite{kara2020robustness}, we studied continuity of infinite horizon discounted cost problem under the convergence of transition kernels. In the following, we first show that the sufficient conditions presented in \cite{kara2020robustness} to guarantee the continuity may not be sufficient for the infinite horizon average cost setup. 

Let the infinite horizon discounted cost under a policy $\gamma$ and the optimal cost be defined by
  \[J_{\beta}(\mathcal{T},\gamma):= E_P^{\mathcal{T},\gamma}\left[\sum_{t=0}^{\infty} \beta^t c(X_t,U_t)\right]\]
  \[ J_{\beta}^*(\mathcal{T}):=\inf_{\gamma\in\Gamma} J_{\beta}(\mathcal{T},\gamma).\]

The following theorem gives sufficient conditions to guarantee the continuity of optimal discounted cost function under weak convergence of transition kernels. 
\begin{theorem} \cite[Theorem 5.2]{kara2020robustness} \label{fully_weak_suff_cont}
Assume the following assumptions hold:
\begin{itemize}
 \item[(i)] $\mathcal{T}_n(\cdot|x_n,u_n)\to\mathcal{T}(\cdot|x,u)$ weakly for every $x\in\mathds{X}$ and $u\in\mathds{U}$ and $(x_n,u_n)\to (x,u)$,
\item[(ii)] $\mathcal{T}(\cdot|x,u)$ is weakly continuous in $(x,u)$, 
\item[(iii)] 
The cost function $c(x,u)$ is continuous and bounded in $\mathds{X}\times\mathds{U}$, 
\item[(iv)] The action space $\mathds{U}$ is compact. 
\end{itemize}
Then $J_\beta(\mathcal{T}_n,\gamma_n^*)\to J_\beta(\mathcal{T},\gamma^*)$, for any initial state $x_0$, as $n\to \infty$.
\end{theorem}

In the following, we show that these conditions in Theorem \ref{fully_weak_suff_cont} may not be sufficient for continuity for average cost problems. The following example shows that even for a control-free setup even if the assumptions of the above theorem hold, infinite horizon average cost function is not continuous. This example also covers the controlled case via using trivial control.
\begin{exmp}
Assume that $x_0=0$ and the transition kernels are given by
\begin{align*}
&\mathcal{T}(\cdot|x)= \delta_{x}(\cdot),\quad \mathcal{T}_n(\cdot|x)= \delta_{x+\frac{1}{n}}(\cdot).
\end{align*}
The cost function is given as
\begin{align*}
c(x)=&\begin{cases} |x| \quad \text{ if } |x|\leq1,\\
1  \quad \text{ if } |x|>1.\end{cases}
\end{align*}
Notice that $\mathcal{T}_n(\cdot|x_n)\to\mathcal{T}(\cdot|x)$ weakly for any $x_n\to x$ and $\mathcal{T}(\cdot|x)$ is weakly continuous. It is easy to see that the cost for $\mathcal{T}$ is 0 as the state always stays at $0$ that is $J_\infty(\mathcal{T})=0$. The cost for $\mathcal{T}_n$ can be calculated as follows:
\begin{align*}
&J_\infty(\mathcal{T}_n)=\lim_{N \to \infty}\frac{1}{N}\bigg(\sum_{k=0}^{n}\frac{k}{n} + \sum_{k=n +1}^{N}1\bigg)\\
&=\lim_{N \to \infty}\frac{1}{N}\bigg(\frac{n+1}{2}+N-n -1\bigg)=1 \neq 0.
\end{align*}
\hfill $\diamond$
\end{exmp}


\subsection{Ergodicity properties of controlled Markov chains}

In the following we will denote the set of all stationary policies by $\Gamma_s$ . For the transitions under some stationary policy $\gamma$, we will use the following notation: 
$\mathcal{T}(\cdot|x,\gamma):=\mathcal{T}(\cdot|x,\gamma(x))$.

We also define the $t$-step transition kernel $\mathcal{T}^t(\cdot|x,\gamma)$ in an iterative fashion as follows:
\begin{align*}
&\mathcal{T}^t(\cdot|x,\gamma):=\int \mathcal{T}(\cdot|x_{t-1},\gamma)\mathcal{T}^{t-1}(dx_{t-1}|x,\gamma),
\end{align*}
where $\mathcal{T}^1(\cdot|x,\gamma)=\mathcal{T}(\cdot|x,\gamma)$. 

We will use the following ergodicity condition for some of our results.
\begin{assumption}\label{geo_erg}
For every stationary policy $\gamma$, the transition kernels $\mathcal{T}$ and $\mathcal{T}_n$ lead to positive Harris recurrent chains and in particular admit invariant measures $\pi_\gamma$ and $\pi_\gamma^n$, and for these invariant measures uniformly for every initial point $x\in\mathds{X}$ we have:
\begin{align*}
&\lim_{t\to\infty}\sup_{\gamma\in\Gamma_s}\|\mathcal{T}^{t}(\cdot|x,\gamma)-\pi_\gamma(\cdot)\|_{TV}= 0\\
&\lim_{t\to\infty}\sup_n\sup_{\gamma\in\Gamma_s}\|\mathcal{T}_n^{t}(\cdot|x,\gamma)-\pi^n_\gamma(\cdot)\|_{TV}= 0.
\end{align*}
\end{assumption}
For the results in this paper, we will make use of Assumption \ref{geo_erg}. However, using \cite[Theorem 3.2]{HeMoRo91}, presented in the appendix as Theorem \ref{erg_conds}, alternative conditions can also be used. Notice that condition Theorem \ref{erg_conds}(i) is the same as Assumption \ref{geo_erg}, thus one can make use of different assumptions through the relations provided in Theorem \ref{erg_conds}.

\subsection{Optimality of stationary policies}

For our continuity and robustness results, it will be instrumental to work with stationary policies. This will be without any loss under mild conditions to be presented in this subsection.  An approach for average cost problems is to make use of average cost optimality equation (ACOE). To work with ACOE one usually needs contraction properties of the transition kernel.  
 The following result provides further alternative sufficient conditions on existence of optimal policies (which turn out to be stationary) for infinite horizon average cost problems.

\begin{assumption} \label{her_stat_assmp}
\begin{itemize}
\item [(A)] The condition f in Theorem \ref{erg_conds} (or any other suitable condition among a-i) holds, 
\item[(B)] The action space $\mathds{U}$ is compact,
\item[(C)] $c(x,u)$ is bounded and continuous in $(x,u)$,
\item[(C')] $c(x,u)$ is bounded and continuous in $u$ for every fixed $x$,
\item[(D)] $\mathcal{T}(\cdot|x,u)$ is weakly continuous in $(x,u)$,
\item [(D')]  $\mathcal{T}(\cdot|x,u)$ is setwise continuous in $u$ for every $x$.
\end{itemize}
\end{assumption}

\begin{theorem}\label{stat_opt}\cite[Corollary 3.6]{her89}
Suppose Assumption \ref{her_stat_assmp}  A, B, and, either C and D, or C' and D', hold. Then $J_\infty(\mathcal{T},\gamma)$ admits an optimal stationary policy.
$\hfill\diamond$
\end{theorem}

For the rest of the paper, we will assume that the optimal policies can be selected from those which are stationary with suitable assumptions and we will denote the family of stationary polices by $\Gamma$.


\subsection{Approximation by finite horizon cost}
We denote the $t$-step finite horizon cost function under a stationary policy $\gamma$ and a transition model $\mathcal{T}$ by $J_t(\mathcal{T},\gamma)$ and the corresponding optimal cost is denoted by $J_t^*(\mathcal{T})$:
\begin{align*}
J_t(\mathcal{T,\gamma})&=\sum_{i=0}^{t-1}E_\gamma^{\mathcal{T}}[c(X_i,U_i)]\\
J_t^*(\mathcal{T})&=\inf_{\gamma\in\Gamma}J_t(\mathcal{T,\gamma}).
\end{align*}
The following result shows that the infinite horizon average cost induced by a stationary policy can be approximated by a finite cost under the same stationary policies with proper ergodicity conditions.

\begin{lemma}\label{apprx}
Under Assumption \ref{geo_erg}, if the cost function $c$ is bounded then for every initial state we have
\begin{align*}
&\sup_{\gamma\in\Gamma} \left|\frac{J_t(\mathcal{T},\gamma)}{t}-J_\infty(\mathcal{T},\gamma)\right|\to 0,\\
&\sup_{\gamma\in\Gamma}\sup_{n} \left|\frac{J_t(\mathcal{T}_n,\gamma)}{t}-J_\infty(\mathcal{T}_n,\gamma)\right|\to 0.
\end{align*}
\end{lemma}
\begin{proof}
We have that $J_\infty(\mathcal{T},\gamma)=\int c(x,\gamma(x))\pi^\gamma(dx).$ Thus, we can write 
\begin{align*}
 &\left|\frac{J_t(\mathcal{T},\gamma)}{t}-J_\infty(\mathcal{T},\gamma)\right|\\
&=\bigg|\frac{1}{t}\sum_{i=0}^{t-1}E_\gamma^{\mathcal{T}}[c(X_i,U_i)]-\int c(x,\gamma(x))\pi^\gamma(dx)\bigg|\\
&\leq \frac{1}{t}\sum_{i=0}^{t-1}\bigg|\int c(x_i,\gamma(x_i)) \mathcal{T}^i(dx_i|x_0,\gamma)-\int c(x,\gamma(x))\pi^\gamma(dx)\bigg|\\
&\quad\leq\frac{1}{t}\sum_{i=0}^{t-1}\|c\|_\infty\|\mathcal{T}^i(\cdot|x_0,\gamma)-\pi^\gamma\|_{TV}.
\end{align*}
We now fix an $\epsilon>0$ and choose a $t_\epsilon<\infty$ such that $\|\mathcal{T}^i(\cdot|x_0,\gamma)-\pi^\gamma\|_{TV}<\epsilon$ for all $i>t_\epsilon$. We also choose another $T_\epsilon$ with $\frac{2t_\epsilon}{t}<\epsilon$ for all $t>T_\epsilon$. With this setup, we have
\begin{align*}
&\frac{1}{t}\sum_{i=0}^{t-1}\|\mathcal{T}^i(\cdot|x_0,\gamma)-\pi^\gamma\|_{TV}\\
&\leq\frac{1}{t}\sum_{i=0}^{t_\epsilon-1}\|\mathcal{T}^i_\gamma(\cdot|x_0)-\pi^\gamma\|_{TV}+\frac{1}{t}\sum_{i=t_\epsilon}^{t}\|\mathcal{T}^i_\gamma(\cdot|x_0)-\pi^\gamma\|_{TV}\\
&\leq \frac{2t_\epsilon}{t}+\epsilon\leq 2\epsilon, \qquad \forall t>T_\epsilon.
\end{align*}
We have shown that for any fixed $\epsilon>0$, we can choose a $T_\epsilon<\infty$, independent of $\gamma$, such that 
\begin{align*}
\left|\frac{J_t(\mathcal{T},\gamma)}{t}-J_\infty(\mathcal{T},\gamma)\right|<\epsilon, \quad \forall t>T_\epsilon.
\end{align*}
Hence the result is complete for $\mathcal{T}$. 

For $\mathcal{T}_n$ the result follows from the same steps since we can again choose such $t_\epsilon$ and $T_\epsilon$ due to the uniformity over $n$ and $\gamma$ in Assumption \ref{geo_erg}. \end{proof}
The next result from \cite[Corollary 4.11]{her89}  shows that the optimal infinite horizon cost can be approximated by an optimal finite horizon cost induced by the same transition kernel. 
\begin{lemma}\label{key_step}
Suppose the cost function $c$ is bounded and either Assumption \ref{her_stat_assmp} A, B, C, D or Assumption \ref{her_stat_assmp} A, B, C', D' hold (for $\mathcal{T}$ and $\mathcal{T}_n$). Then, we have
\begin{align*}
&\lim_{t\to\infty}\left|J_\infty^*(\mathcal{T})-\frac{J_t^*(\mathcal{T})}{t}\right|\to 0,\\
&\lim_{t\to\infty}\sup_n\left|J_\infty^*(\mathcal{T}_n)-\frac{J_t^*(\mathcal{T}_n)}{t}\right|\to 0.
\end{align*}
\end{lemma}



\subsection{Continuity under the convergence of transition kernels}

\begin{theorem}\label{TV_main_cont}
We have that $|J_\infty^*(\mathcal{T}_n)-J_\infty^*(\mathcal{T})|\to 0$, under 
\begin{itemize}
\item[c1.] Assumption \ref{her_stat_assmp} $A$, $B$, $C$ and $D$ if $\mathcal{T}_n(\cdot|x_n,u_n)\to\mathcal{T}(\cdot|x,u)$ weakly for any $(x_n,u_n)\to(x,u)$.
\item[c2.]  Assumption \ref{her_stat_assmp} $A$, $B$, $C'$ and $D'$ if $\mathcal{T}_n(\cdot|x,u_n)\to\mathcal{T}(\cdot|x,u)$ setwise for any $u_n\to u$ for every fixed $x$.
\end{itemize}
\end{theorem}
\begin{proof}
We use the following bound:
\begin{align*}
&|J_\infty^*(\mathcal{T}_n)-J_\infty^*(\mathcal{T})|  \\
& \quad \leq \left|J_\infty^*(\mathcal{T}_n)-\frac{J_t^*(\mathcal{T}_n)}{t}\right| \\
&\quad  +\left|\frac{J_t^*(\mathcal{T}_n)}{t}-\frac{J_t^*(\mathcal{T})}{t}\right|+\left|\frac{J_t^*(\mathcal{T})}{t}-J_\infty^*(\mathcal{T})\right|.
\end{align*}
The first and the last terms above can be made arbitrarily small by choosing $t$ large enough uniformly over $n$ using Lemma \ref{key_step} under suitable assumptions. For the second term, we can use continuity results for finite time problems for the fixed $t$ as the assumptions cover the requirements of  (\cite[Theorem 4.2]{kara2020robustness} or \cite[Section 5.3]{kara2020robustness}).
\end{proof}

\begin{remark}
We note that any condition set provided for continuity of the optimal cost function under setwise convergence of the transition kernels is also a sufficient set of conditions for the continuity of the optimal cost function under total variation convergence of the transition kernels. 
\end{remark}

\section{Robustness to Incorrect Controlled Transition Kernel Models}
In this section, we investigate robustness for infinite horizon average cost problems. We first restate the problem: Consider a MDP with transition kernel $\mathcal{T}_n$, and assume that an optimal control policy for this MDP under the average cost criterion is $\gamma_n^*$, that is
\begin{align*}
\inf_{\gamma\in\Gamma}J_\infty(\mathcal{T}_n,\gamma)=J_\infty(\mathcal{T}_n,\gamma_n^*).
\end{align*}
Now, consider another MDP with transition kernel $\mathcal{T}$ whose the optimal cost denoted by $J_\infty^*(\mathcal{T})$. The question we ask is the following: if the controller does not know the true transition kernel $\mathcal{T}$ and calculates an optimal policy assuming the transition kernel is $\mathcal{T}_n$, then the incurred cost by this policy is $J_\infty(\mathcal{T},\gamma_n^*)$. The focus of this section is to find sufficient conditions such that as $\mathcal{T}_n\to\mathcal{T}$, 
\begin{align*}
J_\infty(\mathcal{T},\gamma_n^*)\to J_\infty(\mathcal{T},\gamma^*).
\end{align*}
The first issue with this question is the following one: assume that the MDP with kernel $\mathcal{T}_n$ admits two different optimal policies $\gamma_n^1$ and $\gamma_n^2$. Although, the cost incurred by these policies under the kernel $\mathcal{T}_n$ are the same, under the kernel $\mathcal{T}$ they may have different cost values. That is, even though we have that 
\begin{align*}
J_\infty(\mathcal{T}_n,\gamma_n^1)=J_\infty(\mathcal{T}_n,\gamma_n^2)=J_\infty^*(\mathcal{T}_n),
\end{align*}
we may have $J_\infty(\mathcal{T},\gamma_n^1)\neq J_\infty(\mathcal{T},\gamma_n^2)$. An example is as follows:  Consider a system with state space $\mathds{X}=[-1,1]$, control action space $\mathds{U}=\{-1,0,1\}$, the cost function $c(x,u)=(x-u)^2$ and the transition models given as 
\begin{align*}
&\mathcal{T}_n(\cdot|x,u)=\frac{1}{2}\delta_1(\cdot)+\frac{1}{2}\delta_{-1}(\cdot)\\
&\mathcal{T}(\cdot|x,u)=\delta_0(\cdot)
\end{align*}
Notice that two optimal policies for $\mathcal{T}_n$ are
\begin{align*}
\gamma_n^1(x)=&\begin{cases} 1 \quad \text{ if } x=1,\\
-1  \quad \text{ if } x=-1,\\
0 \quad \text{ else}.\end{cases}
\gamma_n^2(x)=\begin{cases} 1 \quad \text{ if } x\geq0,\\
-1  \quad \text{ if } x<0.\end{cases}
\end{align*}
However, if the initial point is $x_0=0$, we have that $J_\infty(\mathcal{T},\gamma_n^1)=0\neq1=J_\infty(\mathcal{T},\gamma_n^2)$.

In what follows, we show that under total variation convergence of $\mathcal{T}_n \to \mathcal{T}$, this issue does not cause a problem so that we have $J_\infty(\mathcal{T},\gamma_n^*)\to J_\infty(\mathcal{T},\gamma^*)$ for any stationary optimal policy $\gamma_n^*$. However, under weak or convergence of the transition models, we establish the same result under some particularly constructed optimal polices $\gamma_n^*$, namely we focus on the policies that solve the average cost optimality equation (ACOE).

\subsection{Robustness under weak convergence of transition kernels}

\subsubsection*{The Average Cost Optimality Equation (ACOE)}
Now, we discuss the average cost optimality equation, and we use it for analyzing robustness properties of MDPs under weak convergence of the transition kernels. Define the operator $T: B(\mathds{X})\to B(\mathds{X})$ where $B(\mathds{X})$ denotes the set of bounded and measurable functions on $\mathds{X}$ such that for $v \in B(\mathds{X})$
\begin{align}\label{operator}
Tv(x):=\inf_{u\in\mathds{U}}\bigg(c(x,u)+\int_{\mathds{X}} v(y)\mathcal{T}(dy|x,u)\bigg).
\end{align}
We define the span semi-norm of a function $v\in B(\mathds{X})$ by
\begin{align*}
\mathop{sp}(v):= \sup_x v(x)-\inf_x v(x).
\end{align*}
One can show that the operator defined in (\ref{operator}) is a contraction in $B(\mathds{X})$ under the span-norm with Assumption \ref{geo_erg} or any suitable one from Assumption \ref{erg_conds} \cite[Lemma 3.5]{her89}. Hence, according to the Banach fixed point theorem there exists a fixed point $v^*\in B(\mathds{X})$ such that $\mathop{sp}(Tv^* - v^*)=0$. By the definition of the span-norm, $Tv^*(x)-v^*(x)=j^*$ for a constant $j^*$ for all $x\in \mathds{X}$. That is
\begin{align}\label{acoe}
j^*+v^*(x)=\inf_{u\in\mathds{U}}\bigg(c(x,u)+\int_{\mathds{X}} v^*(y)\mathcal{T}(dy|x,u)\bigg).
\end{align}
This constant $j^*$ is the optimal infinite horizon average cost, and equation (\ref{acoe}) is called the average cost optimality equation (ACOE). For the remainder of this section, we will sometimes use the notation $J_\infty(\mathcal{T},\gamma,x)$ or $J_\infty^*(\mathcal{T},x)$ for the expected average cost or for the expected optimal cost when the process starts from the initial state $x$ in order to emphasize the initial state.

Now, we formalize these observations: 
\begin{theorem}\cite{her89}
Suppose the cost function $c$ is bounded. Under Assumption \ref{geo_erg}, there exists a $\beta<1$ such that the following holds:
\begin{itemize}
\item[(i)] $\mathop{sp}\big(Tv-Tw\big)\leq \beta \mathop{sp}(u-w)$, for any $v,w\in B(\mathds{X})$ where $T$ is the operator defined in (\ref{operator}).
\item[(ii)] Since $T$ is a contraction under the span norm, it admits a fixed point $v^* \in B(\mathds{X})$ such that 
\begin{align*}
j^*+v^*(x)=\inf_{u\in\mathds{U}}\bigg(c(x,u)+\int_{\mathds{X}} v^*(y)\mathcal{T}(dy|x,u)\bigg),
\end{align*}
for some constant $j^*$. 
\item[(iii)] For any initial point $x_0\in \mathds{X}$, the constant $j^*$ defined in $(ii)$ is the optimal infinite horizon average cost for the kernel $\mathcal{T}$, that is
\begin{align*}
j^*=J^*_\infty(\mathcal{T},x_0)=\inf_{\gamma\in\Gamma}J_\infty(\mathcal{T},\gamma,x_0)
\end{align*}
for every $x_0\in \mathds{X}$.
\item[(iv)] If there exists a policy $\gamma^*\in \Gamma$ satisfying the ACOE, then this stationary policy is an optimal policy for the average infinite horizon cost problem; that is, if $\gamma^*$ satisfies
\begin{align*}
j^*+v^*(x)=c(x,\gamma^*(x))+\int_{\mathds{X}} v^*(y)\mathcal{T}(dy|x,\gamma^*(x)),
\end{align*}
then $J_\infty(\mathcal{T},\gamma^*,x_0)=J_\infty^*(\mathcal{T},x_0)$.
\end{itemize}
\end{theorem}






We now state the main result of this section.

\begin{theorem}\label{weak_robust}
We have that \[J_\infty(\mathcal{T},\gamma_n^*,x)\to J_\infty^*(\mathcal{T},x)\]
 for any $x\in\mathds{X}$, where $\gamma_n^*$ is the optimal policy for the transition kernel $\mathcal{T}_n$ that satisfies the ACOE,
under Assumption \ref{her_stat_assmp} $A$, $B$, $C$ and $D$ if $\mathcal{T}_n(\cdot|x_n,u_n)\to\mathcal{T}(\cdot|x,u)$ weakly for any $(x_n,u_n)\to(x,u)$.
\end{theorem}

\begin{proof}
Consider the following two ACOEs for the kernels $\mathcal{T}_n$ and $\mathcal{T}$ with their fixed points $v_n^*$ and $v^*$:
\begin{align}\label{acoes}
j_n^*+v_n^*(x)=\inf_{u\in\mathds{U}}\bigg[c(x,u)+\int v_n^*(y)\mathcal{T}_n(dy|x,u)\bigg]\\
j^*+v^*(x)=\inf_{u\in\mathds{U}}\bigg[c(x,u)+\int v^*(y)\mathcal{T}(dy|x,u)\bigg]
\end{align}

We now show that, for all $x_n \to x$,
\begin{align}\label{fixed_point_conv}
v_n^*(x_n)-v^*(x)\to c
\end{align}
for some constant $c$ with $|c|<\infty$. To show this, we first write
\begin{align*}
&v_n^*(x_n)-v^*(x)\\
&= \big(v_n^*(x_n)- v_n^t(x_n)\big)+\big(v_n^t(x_n)-v^t(x)\big)+\big(v^t(x)-v^*(x)\big)
\end{align*}
where $v_n^t$ and $v^t$ are the results of operator (\ref{operator}) applied to the $0$-function, $t$ times for kernels the $\mathcal{T}_n$ and $\mathcal{T}$. Notice that $v_n^t$ and $v^t$ are the value functions for $t$-step cost problem and by the assumptions (\cite[Theorem 4.4]{kara2020robustness}) we have that $|v_n^t(x_n)-v^t(x)|\to 0$ for every fixed $t$. For the first and the last terms, we use the fact that the operator (\ref{operator}) is a contraction under Assumption \ref{erg_conds} for span semi-norm and hence both terms go to some constants as $t\to\infty$ uniformly for all $n$, that is $v_n^*(x_n)- v_n^t(x_n)\to c_1$ and $v^t(x)-v^*(x)\to c_2$ for some $|c_1|,|c_2|<\infty$. Thus, we have that (\ref{fixed_point_conv}) holds for some $c<\infty$.

Since $\mathds{U}$ is compact, for every $x_n\to x$, $\gamma_n(x_n)$ has a convergent subsequence which converges to say some $u^*\in\mathds{U}$. If we take the limit along this subsequence for (\ref{acoes}), using the assumptions that $\mathcal{T}_n(\cdot|x_n,u_n)\to \mathcal{T}(\cdot|x,u)$ weakly, the fact that $\lim_{n\to\infty}\big(v_n^*(x_n)-v^*(x)\big)=c$, and that $j_n^*\to j^*$ (continuity results from Theorem \ref{TV_main_cont}) we get
\begin{align*}
&\lim_{k}\bigg(j_{n_k}^*+v_{n_k}^*(x_{n_k})\bigg)\\
&=\lim_{k}c(x,\gamma_{n_k}^*(x_{n_k}))+\int v_{n_k}^*(y)\mathcal{T}_{n_k}(dy|x_{n_k},\gamma_{n_k}^*(x_{n_k}))\\
&=j^*+v^*(x)+c=c(x,u^*)+\int v^*(y)\mathcal{T}(dy|x,u^*)+c.
\end{align*}  
Therefore, $u^*$ satisfies the ACOE for the kernel $\mathcal{T}$ and thus, any convergent subsequence of $\gamma_n^*(x_n)$ is an optimal action for $x$ for the kernel $\mathcal{T}$.

Now consider the following operator $\hat{T}_n$, for the kernel $\mathcal{T}$ and the policy $\gamma_n^*$ which is optimal for $\mathcal{T}_n$
\begin{align}\label{hat_T}
\hat{T}_n\hat{v}_n(x)=c(x,\gamma_n^*(x))+\int \hat{v}_n(y)\mathcal{T}(dy|x,\gamma_n^*(x)).
\end{align}
One can show that this operator is also a contraction under span semi-norm and admits a fixed point $\hat{v}_n^*$, such that
\begin{align*}
\hat{j}_n+\hat{v}_n^*(x)=c(x,\gamma_n^*(x))+\int \hat{v}_n^*(y)\mathcal{T}(dy|x,\gamma_n^*(x))
\end{align*}
where $\hat{j}_n=J_\infty(\mathcal{T},\gamma_n^*,x)$ for all $x$. Hence, we need to show that $\hat{j}_n\to j^*$ to complete the proof. To show this, in Appendix \ref{fixed_point_span_conv_app}, we prove that 
\begin{align}\label{fixed_point_span_conv}
\lim_{n\to\infty} \hat{v}_n^*(x_n)-v^*(x)= \hat{c},
\end{align}
for any $x_n\to x$ for some constant $\hat{c}<\infty$.

Now, assume that $\lim_n \hat{j}_n\neq j^*$ and that there exists a subsequence $\hat{j}_{n_k}$ and an $\epsilon>0$ such that $|\hat{j}_{n_k}-j^*|>\epsilon$ for every $k$. We will show that this cannot hold, by establishing the existence of a further subsequence $\hat{j}_{n_{k_l}}$ which converges to $j^*$ in the following.

We first note that $\lim_{n\to\infty} \hat{v}_n^*(x_n)-v^*(x)= \hat{c}$.  Hence, \cite[Theorem 3.5]{Langen81} (or \cite[Theorem 3.5]{serfozo82}) yields that $\int \hat{v}_{n_{k_l}}^*(y)\mathcal{T}(dy|x,\gamma_{n_{k_l}}^*(x))\to \int v^*(y)\mathcal{T}(dy|x,u^*)+\hat{c}$ where $\hat{c}$ also satisfies $\big(\hat{v}_{n_{k_l}}^*(x)-v^*(x)\big)\to \hat{c}$.

Therefore, taking the limit along this subsequence,
\begin{align*}
&\lim_{l\to\infty}\hat{j}_{n_{k_l}} \\
&=\lim_{l\to\infty}c(x,\gamma_{n_{k_l}}^*(x))+\int \hat{v}_{n_{k_l}}^*(y)\mathcal{T}(dy|x,\gamma_{n_{k_l}}^*(x))-\hat{v}_{n_{k_l}}^*(x)\\
&=c(x,u^*)+\int v^*(y)\mathcal{T}(dy|x,u^*)-v^*(x)=j^*.
\end{align*}
This contradicts to $|\hat{j}_{n_k}-j^*|>\epsilon$, hence we conclude that $\hat{j}_n \to j^*$.
\end{proof}

\subsection{Robustness under setwise convergence of transition kernels}

\begin{theorem}\label{setwise_robust}
We have that $J_\infty(\mathcal{T},\gamma_n^*,x)\to J_\infty^*(\mathcal{T},x)$ for any $x\in\mathds{X}$, where $\gamma_n^*$ is the optimal policy for the transition kernel $\mathcal{T}_n$ that satisfies the ACOE,
under Assumption \ref{her_stat_assmp} $A$, $B$, $C$ and $D$ if $\mathcal{T}_n(\cdot|x,u_n)\to\mathcal{T}(\cdot|x,u)$ weakly for any $u_n\to u$.
\end{theorem}

\begin{proof}
The proof follows the similar steps as in the proof of Theorem \ref{weak_robust}. Consider again the following two ACOE for the kernels $\mathcal{T}_n$ and $\mathcal{T}$ with their fixed point $v_n^*$ and $v^*$
\begin{align}\label{acoes_s}
j_n^*+v_n^*(x)=\inf_{u\in\mathds{U}}\bigg[c(x,u)+\int v_n^*(y)\mathcal{T}_n(dy|x,u)\bigg]\\
j^*+v^*(x)=\inf_{u\in\mathds{U}}\bigg[c(x,u)+\int v^*(y)\mathcal{T}(dy|x,u)\bigg].
\end{align}
With the same argument used to establish (\ref{fixed_point_conv}) now using \cite[Theorem 4.8]{kara2020robustness}  one show that $v_n^*(x)-v^*(x)\to c$, for some constant $|c|<\infty$ for all $x$. 

Since $\mathds{U}$ is compact, for every $x$, $\gamma_n(x)$ has a convergent subsequence which converges to say some $u^*\in\mathds{U}$. If we take the limit along this subsequence for (\ref{acoes_s}), using the assumptions that $\mathcal{T}_n(\cdot|x,u_n)\to \mathcal{T}(\cdot|x,u)$ setwise, the fact that $\lim_{n\to\infty}\big(v_n^*(x)-v^*(x)\big)=c$, and that $j_n^*\to j^*$ (continuity results from Theorem \ref{TV_main_cont}) we get
\begin{align*}
&\lim_{k}j_{n_k}^*+v_{n_k}^*(x)\\
&=c(x,\gamma_{n_k}^*(x))+\int v_{n_k}^*(y)\mathcal{T}_{n_k}(dy|x,\gamma_{n_k}^*(x))\\
&=j^*+v^*(x)+c=c(x,u^*)+\int v^*(y)\mathcal{T}(dy|x,u^*)+c.
\end{align*}  
Therefore, $u^*$ satisfies the ACOE for the kernel $\mathcal{T}$ and thus, any convergent subsequence of $\gamma_n^*(x)$ is an optimal action for $x$ for the kernel $\mathcal{T}$.

Now consider the operator $\hat{T}_n$ again,
\begin{align}\label{hat_T_s}
\hat{T}_n\hat{v}_n(x)=c(x,\gamma_n^*(x))+\int \hat{v}_n(y)\mathcal{T}(dy|x,\gamma_n^*(x)).
\end{align}
We write
\begin{align*}
\hat{j}_n+\hat{v}_n^*(x)=c(x,\gamma_n^*(x))+\int \hat{v}_n^*(y)\mathcal{T}(dy|x,\gamma_n^*(x))
\end{align*}
where $\hat{j}_n=J_\infty(\mathcal{T},\gamma_n^*,x)$ for all $x$. Hence, we need to show that $\hat{j}_n\to j^*$ to complete the proof. By replicating the same arguments in Appendix \ref{fixed_point_span_conv_app} for setwise convergence and using \cite[Theorem 20]{royden} (setwise convergence with varying functions), one can prove that 
\begin{align*}
\lim_{n\to\infty} \hat{v}_n^*(x)-v^*(x)= \hat{c},
\end{align*}
for any $x$ for some constant $\hat{c}<\infty$.

Now, assume that $\lim_n \hat{j}_n\neq j^*$ and that there exists a subsequence $\hat{j}_{n_k}$ and an $\epsilon>0$ such that $|\hat{j}_{n_k}-j^*|>\epsilon$ for every $k$. We first note that $\lim_{n\to\infty} \hat{v}_n^*(x)-v^*(x)= \hat{c}$.  Hence, \cite[Theorem 20]{royden} yields that $\int \hat{v}_{n_{k_l}}^*(y)\mathcal{T}(dy|x,\gamma_{n_{k_l}}^*(x))\to \int v^*(y)\mathcal{T}(dy|x,u^*)+\hat{c}$ where $\hat{c}$ also satisfies $\big(\hat{v}_{n_{k_l}}^*(x)-v^*(x)\big)\to \hat{c}$.

Therefore, taking the limit along this subsequence,
\begin{align*}
& \lim_{l\to\infty}\hat{j}_{n_{k_l}} \\
&=\lim_{l\to\infty}c(x,\gamma_{n_{k_l}}^*(x))+\int \hat{v}_{n_{k_l}}^*(y)\mathcal{T}(dy|x,\gamma_{n_{k_l}}^*(x))-\hat{v}_{n_{k_l}}^*(x)\\
&=c(x,u^*)+\int v^*(y)\mathcal{T}(dy|x,u^*)-v^*(x)=j^*.
\end{align*}
This contradicts to $|\hat{j}_{n_k}-j^*|>\epsilon$, hence we conclude that $\hat{j}_n \to j^*$.
\end{proof}

\subsection{Robustness under total variation convergence of transition kernels}

In this section, we will show that for any stationary policy $\gamma_n^*$ that is optimal for $\mathcal{T}_n$, as $\mathcal{T}_n\to\mathcal{T}$ in total variation, we have $J_\infty(\mathcal{T},\gamma_n^*)\to J_\infty^*(\mathcal{T})$ under proper conditions.

We note that, since setwise convergence is less stringent than total variation, it is not surprising that we can establish robustness under total variation convergence of the kernels as well. However, in the previous section, we showed convergence only when the policies were restricted to be among those which solves the ACOE for every point in the state space; in the analysis below, the result will be more general and the policies considered are just required to be optimal, without the requirement that they solve the ACOE for every $x \in \mathds{X}$. This is not vacuous, as the following example shows.

\begin{exmp}
Consider $\mathds{X}=[-1,1]$ and $\mathds{U}=\{0,1,2\}$.

Let the kernels be given in the following form for $n\geq1$:
\begin{align*}
\mathcal{T}_n(\cdot|x,u)&=\bigg(\frac{1}{2}\delta_{\frac{1}{n}}(\cdot)+\frac{1}{2}\delta_{-\frac{1}{n}}(\cdot)\bigg)\mathds{1}_{\{x\geq\frac{1}{n}\}}\\
&+\bigg(\frac{1}{2}\delta_{\frac{1}{n}}(\cdot)+\frac{1}{2}\delta_{-\frac{1}{n}}(\cdot)\bigg)\mathds{1}_{\{x\leq-\frac{1}{n}\}}\\
& +\bigg(\frac{1}{3}\delta_{\frac{1}{n}}(\cdot)+\frac{1}{3}\delta_{-\frac{1}{n}}(\cdot)+\frac{1}{3}\delta_{0}(\cdot)\bigg)\mathds{1}_{\{-\frac{1}{n}< x<\frac{1}{n}\}}\\
\mathcal{T}(\cdot|x,u)=&\delta_{0}(\cdot).
\end{align*}
The cost function is given by:
\begin{align*}
c(x,u)=
\begin{cases}
(x)\mathds{1}_{x\geq0} + 0\mathds{1}_{x<0}  &\text{if } u=0,1\\
3 &\text{if } u=2.
\end{cases}
\end{align*}
One can show that this setup satisfies Assumption \ref{her_stat_assmp} $A$, $B$, $C$ and $D$. With this setup, one (among many others) optimal policy for $\mathcal{T}_n$ when the initial state is $x=-1$ is given by;
\begin{align*}
\gamma_n^*(x)=
\begin{cases}
1  &\text{if } x\leq -\frac{1}{n}\\
0  &\text{if } x\geq \frac{1}{n}\\
2  &\text{otherwise}.
\end{cases}
\end{align*}
When the initial state is $-1$, the cost under this policy is 
\begin{align*}
J_\infty(\mathcal{T}_n,\gamma_n^*)=\lim_{N\to \infty}\frac{1}{N}N\frac{1}{2}\frac{1}{n}=\frac{1}{2n}\to0.
\end{align*}
Therefore the policy $\gamma_n^*$ is indeed optimal for $\mathcal{T}_n$ for large $n$. An  optimal policy for $\mathcal{T}$ is given by $\gamma^*(x)=1$. Thus, the average cost values can be calculated as:
\begin{align*}
& J_\infty(\mathcal{T},\gamma_n^*)=\lim_{N\to\infty}\frac{1}{N}\sum_{t=0}^{N-1}E[c(X_t,\gamma_n^*(X_t))]\\
&=\lim_{N\to\infty}\frac{1}{N}\sum_{t=1}^{N-1}c(0,\gamma_n^*(0))=\lim_{N\to\infty}\frac{1}{N} \sum_{t=1}^{N-1}3=3
\end{align*}
Hence, we have that $$\lim_{n\to\infty}J_\infty(\mathcal{T},\gamma_n^*)=3\neq J_\infty(\mathcal{T},\gamma^*)=0.$$
Notice that for this example, $\gamma_n^*$ would not have been be optimal for $\mathcal{T}_n$ if the initial state were between $-1/n$ and $1/n$ and, in particular, it does not satisfy the ACOE.  
\hfill $\diamond$
\end{exmp}

\begin{theorem}\label{TV_robust}
We have that $|J_\infty(\mathcal{T},\gamma_n^*)-J_\infty^*(\mathcal{T})|\to 0$ for any stationary optimal policy $\gamma_n^*$ for $\mathcal{T}_n$, under Assumption \ref{her_stat_assmp} $A$, $B$, $C'$ and $D'$ if $\mathcal{T}_n(\cdot|x,u_n)\to\mathcal{T}(\cdot|x,u)$ in total variation for any $u_n\to u$ for every fixed $x$.
\end{theorem}

\begin{proof}
We write:
\begin{align*}
&|J_\infty(\mathcal{T},\gamma_n^*)-J_\infty^*(\mathcal{T})|\\
&\qquad\leq |J_\infty^*(\mathcal{T}_n)-J_\infty^*(\mathcal{T})|+ |J_\infty(\mathcal{T},\gamma_n^*)-J_\infty^*(\mathcal{T}_n)|
\end{align*}
the first term goes to by Theorem \ref{TV_main_cont}. For the second term we write
\begin{align*}
&|J_\infty(\mathcal{T},\gamma_n^*)-J_\infty^*(\mathcal{T}_n)|\\
&\leq\left|J_\infty(\mathcal{T},\gamma_n^*)-\frac{J_t(\mathcal{T}_n,\gamma_n^*)}{t}\right| \\
& +\left|\frac{J_t(\mathcal{T}_n,\gamma_n^*)}{t}-\frac{J_t(\mathcal{T},\gamma_n^*)}{t}\right|+\left|\frac{J_t(\mathcal{T},\gamma_n^*)}{t}-J_\infty(\mathcal{T},\gamma_n^*)\right|
\end{align*}
The first and the last terms above again can be made arbitrarily small by choosing $t$ large enough uniformly over $n$ using Lemma \ref{apprx}. For the second term we use \cite[Section A.2]{kara2020robustness} where it is shown that under the stated assumptions $\sup_{\gamma\in\Gamma}|J_t(\mathcal{T}_n,\gamma)-J_t(\mathcal{T},\gamma)|\to 0$. Hence the proof is complete.
\end{proof}

\subsection{Comparison with the literature}
In the most relevant contribution (to our knowledge), \cite{her89}, the following problem is considered: Suppose we have an approximating model for the true kernel $\mathcal{T}$, so that, every time step $t$, an estimate is updated to $\mathcal{T}_t$ and an optimal stationary policy is found with
\begin{align*}
&J_\infty^*(\mathcal{T}_t)=\inf_{\gamma\in \Gamma}\limsup_{N\to\infty}\frac{1}{N}\sum_{i=0}^{N-1}E^{\mathcal{T}_t}\big[c(X_i,\gamma(X_i))\big],
\end{align*}
where $E^{{\mathcal T}_t}$ denotes the expectation with transition kernel $\mathcal{T}_t$. Let $\gamma_t^*$ be the optimal policy for the kernel $\mathcal{T}_t$. It is shown that (\cite[Theorem 5.7]{her89}) if condition f of Theorem \ref{erg_conds} holds for $\mathcal{T}$, $\mathcal{T}_n$ and if $\sup_{x,u}\|\mathcal{T}_t(\cdot|x,u)-\mathcal{T}(\cdot|x,u)\|_{TV}\to 0,$ the following hold:
\begin{itemize}
\item[(i)] $J_\infty(\mathcal{T}_t,\gamma_t^*)\to J_\infty(\mathcal{T},\gamma^*)$ as $t\to\infty$,
\item[(ii)] $J_\infty(\mathcal{T},\gamma_t^*)= J_\infty(\mathcal{T},\gamma^*)$.
\end{itemize} 
Notice that the first item (i) is the continuity problem we study in this paper (Problem 1). Therefore, it can also be shown to hold true with Theorem \ref{TV_main_cont}, which can be proven with  $$\sup_{u}\|\mathcal{T}_t(\cdot|x,u)-\mathcal{T}(\cdot|x,u)\|_{TV}\to 0$$ for fixed $x$. Theorem \ref{TV_main_cont} also states that we can even weaken the total variation convergence of kernels to weak convergence so that it suffices to have  $\mathcal{T}_t(\cdot|x_t,u_t)\to\mathcal{T}(\cdot|x,u)$ weakly for any $(x_t,u_t)\to(x,u)$. For the second item (ii), policies $\gamma_t^*$ (since they are not time-invariant) are not stationary for the model $\mathcal{T}$, however it is shown in \cite{her89} that using $\gamma_t^*$ at time step $t$, the cost 
$$\limsup_{N\to\infty}\frac{1}{N}\sum_{t=0}^{N-1}E^{\mathcal{T}}\big[c(X_t,\gamma_t^*(X_t))\big]$$ is equal to the optimal cost for the true model. 


The following example shows that the total variation convergence of the transition kernels can be too much to ask for deterministic problems. Thus, the relaxation to weak convergence of transition kernels is significant.
\begin{exmp}
Assume that $x_0=0$ and the transition kernels are given by
\begin{align*}
&\mathcal{T}(\cdot|x)= \delta_{0}(\cdot),\quad \mathcal{T}_n(\cdot|x)= \delta_{\frac{1}{n}}(\cdot).
\end{align*}
The cost function is given as
\begin{align*}
c(x)=&\begin{cases} |x| \quad \text{ if } |x|\leq1,\\
1  \quad \text{ if } |x|>1.\end{cases}
\end{align*}
Notice that $\|\mathcal{T}_n(\cdot|x)-\mathcal{T}(\cdot|x)\|_{TV}=2$ for all $n$, however, $\mathcal{T}_n(\cdot|x_n)\to\mathcal{T}(\cdot|x)$ weakly for any $x_n\to x$. Furthermore, $J_\infty^*(\mathcal{T}_n)\to J_\infty^*(\mathcal{T})$.
Hence, although the transition kernels do not converge to each other in total variation, continuity still holds.
\hfill $\diamond$
\end{exmp}

\section{Robustness under Convergence of Transition Kernels without Uniform Ergodicity}\label{SectionEqui}

In this section we show that if the family of optimal policies forms an equicontinuous and stationary family we can guarantee continuity and robustness without requiring a uniform ergodicity over policies and initial points. We first present a supporting lemma from \cite[Theorem 3.5]{Langen81} and \cite[Theorem 3.5]{serfozo82}. 
\begin{lemma}\label{langen}
Suppose $\{\mu_n\}_n \subset \mathcal{P}(\mathds{X})$, where $\mathds{X}$ is a Polish space, converges weakly to some $\mu \in \mathcal{P}(\mathds{X})$. For a bounded real valued sequence of functions $\{f_n\}_n$ such that $\|f_n\|_\infty <C$ for all $n>0$ with $C<\infty$, if $\lim_{n \to \infty}f_n(x_n)=f(x)$ for all $x_n \to x$, i.e. $f_n$ continuously converges to $f$, then
$\lim_{n \to \infty}\int_{\mathds{X}}f_n(x)\mu_n( dx)=\int_{\mathds{X}}f(x)\mu( dx)$.
\end{lemma}

The following result shows that if we restrict the family of policies to an equicontinuous and stationary family we can guarantee continuity and robustness. For the result we do not require a uniform ergodicity assumption as in Assumption \ref{geo_erg}.  

\begin{assumption}\label{conv_assmp}
\begin{itemize}
\item[(i)] For any stationary policy $\gamma$, $\mathcal{T}_n$ and $\mathcal{T}$ lead to positive Harris recurrent chains and in particular admit unique invariant measures $\pi_n^\gamma$ and $\pi^\gamma$.
\item[(ii)] $\{\pi_n^\gamma\}_{\gamma,n}$ and $\{\pi^\gamma\}_\gamma$ are tight.
\item[(iii)] Family of optimal policies $\{\gamma^*_n\}_n$ for $\mathcal{T}_n$, is an equicontinuous family of functions and the optimal policy $\gamma^*$ for $\mathcal{T}$ is continuous.
\end{itemize}
\end{assumption}
\begin{theorem}\label{equ_pol}
Suppose that Assumption \ref{conv_assmp} and Assumption \ref{her_stat_assmp} $B$, $C$, $D$ (for $\mathcal{T}_n$ and $\mathcal{T}$) hold.
Then we have that $J_\infty^*(\mathcal{T}_n)\to J^*(\mathcal{T})$ and $J_\infty(\mathcal{T},\gamma_n^*)\to J_\infty^*(\mathcal{T})$.
\end{theorem}
\begin{proof}
We now use the following bounds. Let $\gamma_n^*$ be optimal for $\mathcal{T}_n$ and $\gamma^*$ be optimal for $\mathcal{T}$. Then, 
\begin{align}\label{usual_bounds}
&|J_\infty(\mathcal{T}_n,\gamma_n^*)-J_\infty(\mathcal{T},\gamma^*)\big|\nonumber\\
&\leq \max \bigg(J_\infty(\mathcal{T}_n,\gamma^*)-J_\infty(\mathcal{T},\gamma^*),J_\infty(\mathcal{T},\gamma_n^*)-J_\infty(\mathcal{T}_n,\gamma_n^*)\bigg),\nonumber\\
&|J_\infty(\mathcal{T},\gamma_n^*)- J_\infty^*(\mathcal{T})|\nonumber\\
&\leq |J_\infty(\mathcal{T},\gamma^*)- J_\infty(\mathcal{T}_n,\gamma^*)|+|J_\infty(\mathcal{T}_n,\gamma_n^*)-J_\infty(\mathcal{T},\gamma_n^*)|.
\end{align}
 Hence, it suffices to show that
\begin{align*}
&\big|J_\infty(\mathcal{T}_n,\gamma^*)-J_\infty(\mathcal{T},\gamma^*)\big|\to0,\\
&\big|J_\infty(\mathcal{T},\gamma_n^*)-J_\infty(\mathcal{T}_n,\gamma_n^*)\big|\to 0.
\end{align*}
First notice that for any policy $\gamma$, $J_\infty(\mathcal{T}_n,\gamma)=\int c(x)\pi^{\gamma}_n(dx)$ and $J_\infty(\mathcal{T},\gamma)=\int c(x)\pi^{\gamma}(dx)$ for all initial states $x_0\in \mathds{X}$ as the chains are positive Harris recurrent. Thus, we only need to show that $\rho(\pi^{\gamma^*}_n, \pi^{\gamma^*})\to 0$ and $\rho(\pi^{\gamma_n^*}, \pi_n^{\gamma_n^*})\to 0$ weakly since $c \in C_b(\mathds{X})$, where $\rho$ metrizes the topology of weak convergence.

For a fixed policy $\gamma^*$, since $\gamma^*$ belongs to an equicontinuous family, we have that $\mathcal{T}_n(\cdot|x_n,\gamma^*(x_n))\to \mathcal{T}(\cdot|x,\gamma^*(x))$ for any $x_n \to x$ and $\mathcal{T}(\cdot|x,\gamma^*(x))$ is weakly continuous in $x$.
Since $\pi_n^{\gamma^*}$ is a tight family, there exists a subsequence $\pi^{\gamma^*}_{n_k}$ such that $\pi^{\gamma^*}_{n_k} \to \pi^*$ weakly for some $\pi^* \in \mathcal{P}(\mathds{X})$. As $\pi^{\gamma^*}_{n_k}$ is the invariant measure for $\mathcal{T}_{n_k}$ we have that for any $f \in C_b(\mathds{X})$
\begin{align*}
\int f(x_1)\mathcal{T}_{n_k}(dx_1|x_0,\gamma^*(x_0))\pi^{\gamma^*}_{n_k}(dx_0)=\int f(x_0)\pi^{\gamma^*}_{n_k}( dx_0)
\end{align*}
using the assumption that $\mathcal{T}_{n_k}(\cdot|x_{n_k},\gamma^*(x_{n_k}))\to\mathcal{T}(\cdot|x,\gamma^*(x))$ weakly for any $x_{n_k}\to x$ and that $\pi^{\gamma^*}_{n_k}\to \pi^*$ by taking the limit $k\to \infty$, Lemma \ref{langen} gives us
\begin{align*}
\int f(x_1)\mathcal{T}( dx_1|x_0,\gamma^*(x_0))\pi^*( dx_0)=\int f(x_0)\pi^*( dx_0).
\end{align*} 
Since $\mathcal{T}$ has a unique invariant measure we can conclude that $\pi^*=\pi^{\gamma^*}$. 

So far we have proved that any converging subsequence of $\pi^{\gamma^*}_{n}$ converges weakly to $\pi^{\gamma^*}$. Now suppose that $\pi^{\gamma^*}_n$ does not converge to $\pi^{\gamma^*}$. Then, there exists an $\epsilon>0$ and a further subsequence $\pi^{\gamma^*}_{n_k}$ such that $\rho(\pi^{\gamma^*}_{n_k},\pi^{\gamma^*})>\epsilon$ for all $k$. But because of the tightness assumption there exists a further subsequence $\pi^{\gamma^*}_{n_{k_l}}$ which converges and using the same arguments above it converges to $\pi^{\gamma^*}$ which leads us to a contradiction and completes the proof.

For the sequence of policies $\gamma_n^*$, using the equicontinuity assumption with Arzel\`a-Ascoli theorem (\cite{Dudley02}), there exists a subsequence $\gamma_{n_k}^*$ converging uniformly to some $\gamma \in\Gamma$. Therefore, $\mathcal{T}_{n_k}(\cdot|x_{n_k},\gamma_{n_k}^*(x_{n_k}))\to \mathcal{T}(\cdot|x,\gamma(x))$ and  $\mathcal{T}(\cdot|x_{n_k},\gamma_{n_k}^*(x_{n_k}))\to \mathcal{T}(\cdot|x,\gamma(x))$ for any $x_{n_k} \to x$.   Hence, using the same steps above, we can show that $\pi_{n_k}^{\gamma^*_{n_k}} \to \pi^{\gamma}$ and $\pi^{\gamma^*_{n_k}} \to \pi^{\gamma}$ weakly. This completes the proof.
\end{proof}

\begin{corollary}
If $\mathds{X}$ is a finite space, $\mathds{U}$ is compact, $\mathcal{T}(x,u)$ is weakly continuous in $u$: 
$c(x,u)$ is continuous in $u$ and under every stationary policy state process $\{x_t\}$ is positive Harris recurrent then we have that $J_\infty^*(\mathcal{T}_n)\to J^*(\mathcal{T})$ and $J_\infty(\mathcal{T},\gamma_n^*)\to J_\infty^*(\mathcal{T})$.
\end{corollary}

\begin{exmp} [Adaptive Control]
Suppose a controlled model is given by the dynamics
\begin{align*}
x_{t+1}=Ax_t+Bu_t+w_t
\end{align*}
where $w_t$ is a i.i.d. Gaussian noise process. Assume that $A$ and $B$ are unknown by the controller. However, the controller can estimate $A$ and $B$ in a consistent way so that $A_n\to A$ and $B_n\to B$; with many results reported in the literature \cite{aastrom2013adaptive,kumar2015stochastic,goodwin1981discrete,goodwin2014adaptive}. Then, building on Section \ref{ornek1}(i) and \ref{ornek1}(vi), we have $\mathcal{T}_n(\cdot|x_n,u_n)\to\mathcal{T}(\cdot|x,u)$ weakly for any $x_n\to x$ and $u_n \to u$. If further we have that the step-wise cost function is in the form $c(x,u)=x^TQx+u^TRu$, then the optimal policies are linear and also equicontinuous if the model is controllable \cite{bertsekas}. Furthermore, since the noise is Gaussian, the chain is Lebesgue irreducible and thus there exists a unique invariant measure and it can be reached from any initial point. Hence, the assumptions of Theorem \ref{equ_pol} are satisfied and the continuity and robustness can be established for this example. In other words;  the optimal cost for the estimates $A_n$ and $B_n$, $J_\infty^*(\mathcal{T}_n)$, converges to optimal cost for $A$ and $B$, $J_\infty^*(\mathcal{T})$. Furthermore, if we apply the optimal policies $\gamma_n^*$ designed for $A_n$ and $B_n$, to the true model ($A$, $B$), we get $J_\infty(\mathcal{T},\gamma_n^*)\to J_\infty^*(\mathcal{T}).$ \hfill $\diamond$
\end{exmp}

\section{Application to Data-Driven Model Learning}\label{secLearning}

\subsection{Data-Driven Model Learning and Robustness}

For unknown probability measures, a common practice is to learn them via test inputs
or empirical observations.  Let $\{(X_i), i \in \mathbb{N} \}$ be an $\mathbb{X}$-valued i.i.d random variable
sequence generated according to some distribution $\mu$. 

Defining for every (fixed) Borel $B \subset \mathbb{X}$, and $n \in
\mathbb{N}$, the empirical occupation measures
\[
\mu_n(B)=
\frac{1}{n}\sum_{i=1}^{n} 1_{\{X_i \in B\}},
\]
one has $\mu_n(B) \to \mu(B)$ almost surely (a.s$.$) by the strong law
of large numbers. Also, $\mu_n \to \mu$ weakly with probability one (\cite{Dudley02}, Theorem 11.4.1).

However, $\mu_n$ can not converge to $\mu$ in total variation, in general. On the other hand, if we know that $\mu$ admits a density, we can find estimators to estimate $\mu$ under total variation \cite{Devroye85}. As discussed above, the empirical averages converge almost surely. By a similar reasoning, for a given bounded measurable function $f$, $\int \mu_n(dx) f(x)$ converges to $\int \mu(dx) f(x)$. This then also holds for any {\it countable} collection of functions, $f_1, f_2, \cdots$. A relevant question is the following: Can one ensure uniform convergence (over a family of functions) with arbitrary precision by only guaranteeing convergence for a finite or countably infinite collection of functions? This entails the problem of covering a family of functions with arbitrarily small neighborhoods of finitely many functions under an appropriate distance metric. The answer to this question is studied by the theory of {\it empirical risk minimization}: In the learning theoretic context when one tries to estimate the source distribution, the convergence of optimal costs under $\mu_n$ to the cost optimal for $\mu$ is called the {\it consistency of empirical risk minimization} \cite{Vap00}. In particular, if the following uniform convergence holds,
\begin{align}\label{unif_conv}
\lim_{n \to \infty} \sup_{f \in {\cal F}} \bigg|\int f(x) \mu_n(dx) - \int f(x) \mu(dx) \bigg| =0,
\end{align}
for a class of measurable functions ${\cal F}$, then ${\cal F}$ is called a {\em $\mu$-Glivenko-Cantelli class} \cite{Dudley}. If the class ${\cal F}$ is $\mu$-Glivenko-Cantelli for every $\mu$, it is called a {\it universal Glivenko-Cantelli} class. One example of a universal Glivenko-Cantelli family of real functions on $\mathbb{R}^N$ is the
family $\{f:\, \|f\|_{BL}\le M\}$ for some $0<M<\infty$, where
$\|f\|_{BL}=\|f\|_\infty + \sup_{x_1\neq x_2}\frac{|f(x_1)-f(x_2)|}{|x_1-x_2|}$ ( \cite{Dudley}). For related characterizations and further examples, see \cite{raginsky2013empirical,VanHandel,dudley1999uniform}.


For our analysis, we will mainly make use of the fact that the empirical occupation measures $\mu_n$ converge weakly almost surely to the true underlying measure $\mu$.

The above discussion is for the i.i.d. observations and for fixed probability measures. In the case of unknown dynamics or unknown transition models, for every different $(x,u)\in \mathds{X}\times \mathds{U}$ pair, $\mathcal{T}(\cdot|x,u)$ is a different probability measure, and hence there are possibly uncountably many unknown probability measures. Thus, in this section we focus on some special settings.

Let $\mathcal{T}(\cdot|x,u)$ be a transition kernel given previous state and action variables $x\in \mathds{X},u\in\mathds{U}$, which is unknown to the decision maker (DM). Suppose the DM builds a model for the transition kernels, $\mathcal{T}_n(\cdot|x,u)$, for all possible $x\in \mathds{X}, u \in \mathds{U}$ by collecting training data $\{x_t, u_t, \quad t \leq n\}$ from the evolving system. 

The question we are interested in is that, do we have that the cost calculated under $\mathcal{T}_n$ converges to the true cost (i.e., do we have that the cost obtained from applying the optimal policy for the empirical model converges to the true cost as the training length increases)? We will refer to this property as being {\it asymptotically robust under empirical learning}. In the following, we provide two setups. 

\begin{theorem}
Consider the example in Section \ref{ornek1}(vi): Let a controlled model be given as
\[x_{t+1}=F(x_t,u_t,w_t),\] where $\{w_t\}$ is an i.i.d. noise process. Suppose that $F(x,u,\cdot): \mathds{W}\to \mathds{X}$ is invertible for all fixed $(x,u)$ and $F(x,u,w)$ is continuous and bounded on $\mathds{X}\times\mathds{U}\times\mathds{W}$. If we construct the empirical measures for the noise process $w_t$  such that for every (fixed) Borel $B \subset \mathds{W}$, and for every $n \in\mathds{N}$, the empirical occupation measures are
\begin{align*}\label{inverse_emp}
\mu_n(B)=\frac{1}{n}\sum_{i=1}^{n} \mathds{1}_{\{F^{-1}_{x_{i-1},u_{i-1}}(x_i) \in B\} }
\end{align*}
where $F^{-1}_{x_{i-1},u_{i-1}}(x_i)$ denotes the inverse of $F(x_{i-1},u_{i-1},w): \mathds{W}\to \mathds{X}$ for given $(x_{i-1},u_{i-1})$. 
Using the noise measurements, we construct the empirical transition kernel estimates for any $(x_0,u_0)$ and Borel $B$ as
\begin{align*}
\mathcal{T}_n(B|x_0,u_0)=\mu_n(F^{-1}_{x_0,u_0}(B)).
\end{align*}
Then, under empirical learning, by Theorem \ref{weak_robust}, asymptotic robustness under empirical learning holds.
\end{theorem}

\begin{proof}
By the analysis in Section \ref{ornek1}(vi), we have that \[\mathcal{T}_n(\cdot|x_n,u_n)\to\mathcal{T}(\cdot|x,u),\] weakly for any $(x_n,u_n) \to (x,u)$ almost surely. Furthermore, additionally let \[F(x_t,u_t,w_t) = G(x_t,u_t) + w_t\] where (in addition to the assumed regularity conditions on $F$) we also have that $G$ has a bounded range and that $w_t$ has a density which is positive everywhere. Then, by \cite[Example 2.2]{mcdonald2020exponential} it follows that Theorem \ref{erg_conds}(h) holds. 
This ensures that under any stationary policy $x_t$ is positive Harris recurrent and also geometrically ergodic.
We thus conclude that under empirical learning, by Theorem \ref{weak_robust}, the loss due to an incorrect initial modelling error is zero.
\end{proof}

\begin{theorem}
Suppose we are given the following dynamics for finite state space, $\mathds{X}$, and finite action space, $\mathds{U}$,
\[ x_{t+1}= f(x_t,u_t,w_t)\]
where $\{w_t\}$ is an i.i.d.noise process and the noise model is unknown.
 Suppose again that there is an initial training period so that under some policy, every $x,u$ pair is visited infinitely often if training were to continue indefinitely, but that the training ends at some finite time. Let us assume that, through this training, we empirically learn the transition dynamics such that for every (fixed) Borel $B \subset \mathds{X}$, for every $x\in \mathds{X}$, $u \in \mathds{U}$ and $n \in\mathds{N}$, the empirical occupation measures are
\[
\mathcal{T}_n(B|x_0=x,u_0=u)=
\frac{\sum_{i=1}^{n} 1_{\{X_i \in B,X_{i-1}=x,U_{i-1} = u\}}}{\sum_{i=1}^{n} 1_{\{X_{i-1}=x,U_{i-1} = u\}}}.
\]
Then we have that $J_\beta^*(\mathcal{T}_n) \to J_\beta^*(\mathcal{T})$ and $J_\beta(\mathcal{T},\gamma_n^*)\to J_\beta^*(\mathcal{T})$, 
where $\gamma_n^*$ is the optimal policy designed for $\mathcal{T}_n$. 
Then, asymptotic robustness under empirical learning holds.
\end{theorem}
\begin{proof} By \cite[Corollary 5.1]{kara2020robustness}, $\mathcal{T}_n(\cdot|x,u) \to \mathcal{T}(\cdot|x,u)$ weakly for every $x \in \mathds{X}$, $u \in \mathds{U}$ almost surely (by law of large numbers). Since the spaces are finite, we also have $\mathcal{T}_n(\cdot|x,u) \to \mathcal{T}(\cdot|x,u)$ under total variation. 

Suppose further that we have $\mathcal{T}(\cdot|x,u) > 0$ for every $x,u$. Then, for large enough $n$, we will have uniform ergodicity by Theorem \ref{erg_conds}. As a result, Theorem \ref{TV_robust} will apply and we will have consistency under empirical learning.
\end{proof}


\subsection{Adaptive Learning}
Suppose we adaptively learn the transition kernel so that at time $n$ we have an estimated kernel ${\cal T}_n$, and accordingly apply an optimal policy for some time period after $n$ based on our improved model. In the following, we will build on Theorem \ref{weak_robust}, Theorem \ref{setwise_robust} or Theorem \ref{TV_robust} to arrive at robustness (and asymptotic consistency) of such an adaptive policy. Let $T_1, T_2, \cdots$ be a sequence of increasing integers that satisfy
\begin{equation}
  \label{eq_kopt2}
\lim_{k\to \infty} \frac{\sum_{l=1}^k T_l}{T_k} =1.
\end{equation}
Suppose that we apply the control policy $\gamma_{n_k}$, which is optimal for ${\cal T}_{n_k}$ with $n_k = \sum_{l=1}^k T_l,$ from $n_{k}$ until $n_{k+1}$. Notice that the length of the time intervals we update the policy grows to $\infty$ with this setup. Call this adaptive policy $\tilde{\gamma}$. Since the cost is bounded and the rate of convergence to the invariant measure is uniform over policies and over initial states $x \in \mathds{X}$ under the conditions of ergodicity in Theorem \ref{erg_conds}, we have that
\begin{align*}
& J_\infty(\mathcal{T},\tilde{\gamma},x) = \lim_{T \to \infty} \frac{\sum_{t={0}}^{T-1} E_x[c(x_t,\tilde{\gamma}_{t}(x_t))]}{T} \\
& = \lim_{k \to \infty} \frac{\sum_{t={n_k}}^{n_{k+1}} E_x[c(x_t,\gamma_{n_k}(x_t))]}{n_{k+1}}  = \lim_{k \to \infty} J_\infty(\mathcal{T},\gamma_{n_k},x).
\end{align*}
In the above, the last two steps follow from the fact that $n_{k+1}-n_k \to \infty$ and $n_{k+1}\to\infty$ at the same rate by (\ref{eq_kopt2}) and that $\gamma_{n_k}$ is optimal from time $n_k$ to $n_{k+1}$. By the results we have established, if ${\cal T}_{n_k} \to {\cal T}$ under any of the senses established by Theorem \ref{weak_robust}, Theorem \ref{setwise_robust} or Theorem \ref{TV_robust}, we will have a strong form of robustness: the cost to due an incorrect initial modeling error will be zero. 

\section{Conclusion}
We studied regularity properties of optimal cost for the average infinite horizon setups on the space of transition kernels, and applications to robustness of optimal control policies designed for an incorrect model applied to an actual system. We applied our results to adaptive control and data-driven learning and established empirical consistency results.

\appendix

\section{Proof of Equation (\ref{fixed_point_span_conv})}\label{fixed_point_span_conv_app}
We first define the following operator $\hat{T}_{n,z}$ for some fixed $z\in \mathds{X}$ by
\begin{align*}
\hat{T}_{n,z} v(x):= \hat{T}_nv(x)-\hat{T}_nv(z)
\end{align*}
where $\hat{T}_n$ is as in (\ref{hat_T}).
 We write
\begin{align*}
&\hat{v}_n^*(x_n)-v^*(x)= \big(\hat{v}_n^*(x_n)-\hat{T}_{n,z}^kv^*(x_n)\big) \\
& \quad +\big(\hat{T}_{n,z}^kv^*(x_n)-v^*(x)\big)
\end{align*}
where $\hat{T}_{n,z}^k$ is the operator $\hat{T}_{n,z}$ applied $k$ consecutive times. We note that $\hat{T}_{n,z}$ is also a contraction under the $span$ semi-norm. Hence, the first term converges to some $\hat{c}_1$ as $k\to \infty$ uniformly for all $n$ since $\hat{T}_{n,z}$ is a contraction uniformly for all $n$ under the span seminorm and its fixed point is $\hat{v}_n^*$. 

For the second term, we wish to show that $\big(\hat{T}_{n,z}^kv^*(x_n)-v^*(x)\big)\to -v^*(z)$ as $n\to\infty$ for every fixed $k<\infty$ for all $x_n\to x$. We prove this by induction.
For $k=1$, we have
\begin{align*}
\hat{T}_{n,z}v^*(x_n)=&c(x_n,\gamma_n^*(x_n))+\int v^*(y)\mathcal{T}(dy|x_n,\gamma_n^*(x_n))\\
&-\bigg(c(z,\gamma_n^*(z))+\int v^*(y)\mathcal{T}(dy|z,\gamma_n^*(z))\bigg).
\end{align*}
Now assume that $|\hat{T}_{n_k,z}v^*(x_{n_k})-v^*(x)+v^*(z)|>\epsilon$ for some $\epsilon >0$ for every $k$ along some subsequence $n_k$. We know that there exists a further subsequence, say $n_{k_l}$ along which $\gamma^*_{n_{k_l}}(x_{n_{k_l}})\to u^*$ and $\gamma^*_{n_{k_l}}(z)\to u^{**}$ for some $u^*, u^{**}\in \mathds{U}$ where $u^*$ is an optimal action for the state $x$ for the kernel $\mathcal{T}$ and $u^{**}$ is optimal for $z$.  If we take the limit along this subsequence, i.e.
\begin{align*}
&\lim_{l\to \infty} \hat{T}_{n_{k_l},z}v^*(x_{n_{k_l}}) \\
&=\lim_{l\to\infty}\bigg(c(x_{n_{k_l}},\gamma_{n_{k_l}}^*(x_{n_{k_l}}))+\int v^*(y)\mathcal{T}(dy|x_{n_{k_l}},\gamma_{n_{k_l}}^*(x_{n_{k_l}}))\bigg)\\
&\qquad-\lim_{l\to\infty}\bigg(c(z,\gamma_{n_{k_l}}^*(z))+\int v^*(y)\mathcal{T}(dy|z,\gamma_{n_{k_l}}^*(z))\bigg)\\
& =c(x,u^*)+\int v^*(y)\mathcal{T}(dy|x,u^*)\\
&\qquad- \bigg(c(z,u^{**})+\int v^*(y)\mathcal{T}(dy|z,u^{**})\bigg)\\
&=v^*(x)+j^*-v^*(z)-j^*=v^*(x)-v^*(z).
\end{align*}
Hence we reach a contradiction and it must be that $\hat{T}_{n,z}v^*(x_n)-v^*(x)\to -v^*(z)$. Now assume that the claim holds for $k$ that is $\big(\hat{T}_{n,z}^kv^*(x_n)-v^*(x)\big)\to -v^*(z)$.
\begin{align*}
& \hat{T}^{k+1}_{n,z}v^*(x_n) =  c(x_n,\gamma_n^*(x_n))+\int \hat{T}_{n,z}^kv^*(y)\mathcal{T}(dy|x_n,\gamma_n^*(x_n)) \\
&\qquad  -\bigg(c(z,\gamma_n^*(z))+\int \hat{T}_{n,z}^kv^*(y)\mathcal{T}(dy|z,\gamma_n^*(z))\bigg).
\end{align*}
If we take a subsequence, indexed by say $m$, along which $\gamma_m^*(x_m)\to u^*$ and $\gamma_m^*(z)\to u^{**}$ where $u^*$ is an optimal action for $x$ and $u^{**}$ is optimal for $z$, then by taking the limit along this subsequence
\begin{align*}
&\lim_{m\to\infty}\hat{T}^{k+1}_mv^*(x_m)\\
&=\lim_{m\to\infty}c(x_m,\gamma_m^*(x_m))+\int \hat{T}_m^kv^*(y)\mathcal{T}(dy|x_m,\gamma_m^*(x_m))\\
&\quad-\lim_{m\to\infty}c(z,\gamma_m^*(z))+\int \hat{T}_m^kv^*(y)\mathcal{T}(dy|z_m,\gamma_m^*(z))\\
&=c(x,u^*)+\int \big(v^*(y)-v^*(z)\big)\mathcal{T}(dy|x,u^*)\\
&\quad-\bigg(c(z,u^{**})+\int \big(v^*(y)-v^*(z)\big)\mathcal{T}(dy|z,u^{**})\bigg)\\
&=v^*(x)-v^*(z).
\end{align*}
Hence, a similar contradiction argument, we used for the case $k=1$ yields that $\big(\hat{T}_{n,z}^kv^*(x_n)-v^*(x)\big)\to -v^*(z)$ as $n\to\infty$ for every fixed $k<\infty$ for all $x_n\to x$. Thus, we have that 
\begin{align*}
&\hat{v}_n^*(x_n)-v^*(x)= \big(\hat{v}_n^*(x_n)-\hat{T}_{n,z}^kv^*(x_n)\big)\\
&+\big(\hat{T}_{n,z}^kv^*(x_n)-v^*(x)\big)\to \hat{c}_1-v^*(z):=\hat{c}.
\end{align*}

\section{Ergodicity conditions on Markov Chains}
The following theorem is stated here for easy reference.

\begin{theorem}\label{erg_conds}\cite[Theorem 3.2]{HeMoRo91}
Consider the following.
\begin{itemize}
\item [a.] There exists a state $x^*\in\mathds{X}$ and a number $\beta>0$ such that $\mathcal{T}(\{x^*\}|x,\gamma)\geq\beta$, for all  $x\in\mathds{X},\gamma\in\Gamma_s$.
\item[b.] There exists a positive integer $t$ and a non-trivial measure $\mu$ on $\mathds{X}$ such that $\mathcal{T}^t(\cdot|x,\gamma)\geq\mu(\cdot)$ for all $x\in\mathds{X},\gamma\in\Gamma_s$.
\item [c.] For each $\gamma\in\Gamma_s$, the transition kernel $\mathcal{T}(dy|x,\gamma)$ has a density $p(y|x,\gamma)$ with respect to a sigma-finite measure $m$ on $\mathds{X}$, and there exist $\epsilon>0$ and $C\in\mathcal{B}(\mathds{X})$ such that $m(C)>0$ and $p(y|x,\gamma)\geq\epsilon$ for all $y\in C, x\in \mathds{X}, \gamma\in\Gamma_s$.
\item[d.] For each $\gamma\in\Gamma_s$, $\mathcal{T}(dy|x,\gamma)$ has a density $p(y|x,\gamma)$ with respect to a sigma-finite measure $m$ on $\mathds{X}$, and $p(y|x,\gamma)\geq p_0(y)$ for all $x,y\in\mathds{X}$, $\gamma\in\Gamma_s$, where $p_0$ is a non-negative measurable function with $\int p_0(y)m(dy)>0$.
\item[e.] There exists a positive integer $t$ and a measure $\mu$ on $\mathds{X}$ such that $\mu(\mathds{X})<2$ and $\mathcal{T}^t(\cdot|x,\gamma)\leq\mu(\cdot)$ for all $x\in\mathds{X}, \gamma\in\Gamma_s$.
\item[f.] There exists a positive integer $t$ and a positive number $\beta<1$ such that $\|\mathcal{T}^t(\cdot|x,\gamma)-\mathcal{T}^t(\cdot|x',\gamma)\|_{TV}\leq 2\beta$ for all $x,x'\in\mathds{X}, \gamma\in\Gamma_s$.
\item[g. ] There exists a positive integer $t$ and a positive number $\beta$ for which the following holds: For each $\gamma\in\Gamma_s$, there is a probability measure $\mu_\gamma$ on $\mathds{X}$ such that $\mathcal{T}^t(\cdot|x,\gamma)\geq \beta\mu_\gamma(\cdot)$ for all $x\in \mathds{X}$.
\item[h. ] There exist positive numbers $c$ and $\beta$, with $\beta<1$, for which the following holds: For each $\gamma\in\Gamma_s$, there is a probability measure $p_\gamma$ on $\mathds{X}$ such that $\|\mathcal{T}^t(\cdot|x,\gamma)-p_\gamma(\cdot)\|_{TV}\leq c\beta^t$ for all $x\in\mathds{X}, t\in\N$.
\item[i. ] The state process is uniformly ergodic such that 
$\lim_{t\to\infty}\|\mathcal{T}^t(\cdot|x,\gamma)-p_\gamma(\cdot)\|_{TV}=0$
uniformly in  $x\in\mathds{X}$  and $\gamma\in\Gamma_s$.
\end{itemize}
The conditions above are related as follows:
\begin{align*}
&a \to b\\
&e\to f\\
&c\to d \to b  \to f \leftrightarrow g \leftrightarrow h \leftrightarrow i.
\end{align*}
\end{theorem}

\bibliographystyle{plain}
\bibliography{references_acc,SerdarBibliography_acc,SerdarBibliography,references,AliBibliography}

\end{document}